\documentclass[conference]{IEEEtran}

\if CLASSOPTIONcompsoc
  \usepackage[nocompress]{cite}
\else
  \usepackage{cite}
\fi

\usepackage[utf8x]{inputenc}
\usepackage{listings}
\usepackage{color}
\usepackage{epsfig}
\usepackage{url}
\usepackage{amsmath}
\usepackage{cite}
\usepackage{graphicx}
\usepackage{color}
\usepackage[bookmarks=false]{hyperref}
\usepackage{amssymb,amsmath,dsfont}
\usepackage{amsthm}
\usepackage[linesnumbered,lined, ruled]{algorithm2e}
\usepackage{epstopdf}
\usepackage{lipsum}
\usepackage{bbm}
\usepackage{tablefootnote}

\usepackage{tikz}
\usetikzlibrary{automata,arrows,positioning,calc}

\usepackage{multirow}

\usepackage{breqn}
\SetAlFnt{\footnotesize}

\usepackage{graphics}

\newtheorem{theorem}{Theorem}

\newtheorem{proposition}{Proposition}

\theoremstyle{remark}

\usepackage{csquotes}
\usepackage{amsmath} 
\usepackage{amssymb}
\usepackage{subcaption}
\usepackage{graphicx}
\usepackage{accents}
\usepackage{mathtools}
\usepackage[normalem]{ulem}

\SetCommentSty{mycommfont}

\graphicspath{figures/}

\usepackage{mathtools}
\DeclarePairedDelimiter\floor{\lfloor}{\rfloor}
\DeclarePairedDelimiter\ceil{\lceil}{\rceil}

\SetKwRepeat{Repeat}{repeat}{until}
\SetKwRepeat{Forever}{repeat}{forever}
\SetKwRepeat{On}{on}{end}
\SetNlSty{bfseries}{\color{black}}{}

\begin{document}

\title{Improved Credit Bounds for the Credit-Based Shaper in Time-Sensitive Networking}


\author{\IEEEauthorblockN{Ehsan Mohammadpour, Eleni Stai, Jean-Yves Le Boudec\\}
	\IEEEauthorblockA{\'Ecole Polytechnique F\'ed\'erale de Lausanne, Switzerland\\
		$\{$firstname.lastname$\}$@epfl.ch}}


\maketitle

\begin{abstract}
In Time-Sensitive Networking (TSN), it is important to formally prove per flow latency and backlog bounds. To this end, recent works apply network calculus and obtain latency bounds from service curves. The latency component of such service curves is directly derived from upper bounds on the values of the credit counters used by the Credit-Based Shaper (CBS), an essential building-block of TSN. In this paper, 
we derive and formally prove credit upper bounds for CBS, which improve on existing bounds. 

\end{abstract} 

\IEEEpeerreviewmaketitle

\section{Introduction}
\label{sec:intro}

Time-Sensitive Networking (TSN) is an IEEE 802.1 working group that defines networking mechanisms for supporting real-time data flows with latency guarantees and zero packet loss \cite{_time-sensitive_task}. TSN targets applications in avionics, automotive, industrial networks, etc., where data loss or latency violation causes catastrophic damage. One of the main building-blocks of TSN is the Credit-Based Shaper (CBS), 
 which provides rate allocation for a number of priority classes, called Audio-Video Bridging (AVB) classes, using a credit mechanism (Section \ref{sec:sm}). 
 In recent studies \cite{mohammadpour_latency_2018,daigmorte_modelling_2018,zhao_timing_analysis_2018}, latency and backlog bounds in TSN are derived by using network calculus and service curve characterizations of CBS \cite{le_boudec_network_2001}. The latency parameters of such service curves is directly derived from credit upper bounds, i.e. bounds on the values of the credit counters of CBS. Two sets of results were published for credit upper bounds. The former, by J. De Azua and M. Boyer, (``J-bounds", \cite{azua_complete_2014}) applies only to the case with two AVB classes and its proof is not fully described. The latter set of bounds, by H. Daigmorte et al (``H-bounds", \cite{daigmorte_modelling_2018}), applies to any number of AVB classes and is formally proven. For the top priority AVB class, J- and H-bounds are identical. For the second priority class, J-bounds are generally smaller than H-bounds, but not always. For third and lower priority classes, only H-bounds are available.

In this paper, we derive and formally prove credit upper bounds for CBS with any number of AVB classes, which improve on both sets of existing bounds. Specifically, our bound is the same as J- and H-bounds for the top priority AVB class. For the second priority class, our bound is lower than the H-bound and is lower than or equal to the J-bound, depending on maximum packet lengths. For all other priority classes, our bounds are lower than the only available bounds, namely the H-bounds.
Moreover, we formally show that our credit bounds are tight for the two classes of highest priority, which is the first tightness result presented in the literature. In Section \ref{sec:eval}, we perform numerical evaluations and show that the improvement in latency guarantees is significant.




\section{System Model \& Existing Credit Bounds}
\label{sec:sm}
We assume a TSN scheduler with the following elements: 1) A set of queues representing a set of classes including, in decreasing priority, one Control-Data Traffic (CDT) class, $p$ AVB classes 1, 2, 3, ..., $p$, and a set of Best Effort (BE) classes; 
2) A set of gates, one per queue, such that if a gate is closed, the corresponding queue cannot transmit. 
A Gate-Control List (GCL) contains the information of the opening/closing times of gates. 
Moreover, there are several integration policies that determine the preemption or non-preemption of CDT over the rest of the classes \cite{daigmorte_modelling_2018}. The analysis in this paper is valid for all integration policies; 
3) A set of CBSs, one per AVB queue, to control the allocated rate of each AVB class. The CBS of an AVB class $i$ has two parameters: the \emph{idle slope}, $I_i>0$, and the \emph{send slope}, $S_i<0$. Note that $I_i-S_i=c$, where $c$ is the line transmission rate. The idle slope is interpreted as the rate guaranteed to class $i$ and therefore, the condition $\sum_{i=1}^{p}I_i<c$ is assumed to hold. 
%
Packets are scheduled according to the following rules (we repeat here the description in \cite{daigmorte_modelling_2018}):\\
\textbf{R1:} If the transmission line is free, the scheduler transmits a packet of the highest priority class that satisfies all the conditions: 1) it has a positive backlog; 2) its gate is open; 3) it has a nonnegative credit if it is an AVB class. \\
\textbf{R2:} The credit of the AVB class $i$ reduces linearly with rate the send slope, $S_i$, if $i$ transmits.\\
\textbf{R3:} The credit of the AVB class $i$ increases linearly with rate $I_i$, when the following conditions hold simultaneously for class $i$: 1) its gate is open; 2) it has a positive backlog; and 3) other AVB or BE classes are transmitting.\\
\textbf{R4:} The credit of an AVB class remains constant, if the corresponding gate is closed, and during any additional overhead in the case of preemption mode for CDT.\\
\textbf{R5:} When class $i$ has a positive credit and its backlog becomes zero, the credit is set to zero; this is the credit reset. If the credit is negative and the backlog becomes zero, the credit increases with rate $I_i$ until the zero value.\\
Let $V_i(t)$ denote the value of the credit counter for AVB class $i$ at time $t\geq 0$. We assume that the system is idle at time $0$ and $V_i(0)=0$. The function $V_i()$ may take positive or negative values and is continuous, except at  credit reset times, which, by R5, may occur only when the queue of class $i$ becomes empty. At all other times it is linearly increasing, decreasing or constant.
The J-bounds in \cite{azua_complete_2014} apply to the case $p=2$, as:
\begin{align}
\label{eq:boyer_avb_A}&  V_1(t) \leq \bar{L} \frac{I_1}{c}:= V^{\text{max},\mathrm{J}}_1, \\
\label{eq:boyer_avb_B}&  V_2(t) \leq \frac{I_2}{c} \Big(L_{BE}+L_1+ \bar{L} \frac{I_1}{-S_1}\Big) := V^{\text{max},\mathrm{J}}_2,
\end{align}
\noindent where $L_i$ and $L_{BE}$ are maximum packet lengths of AVB class $i$ and BE, and $\bar{L} = \max(L_2,L_{BE})$ 
The proof of Eq. \eqref{eq:boyer_avb_B} in \cite{azua_complete_2014} is not completely formal (in contrast, we provide formal proofs for our improved credit bounds).

The H-bounds in \cite{daigmorte_modelling_2018} apply to any value of $p$ 
 and give, for $i=1,...,p$: 
\begin{align}\label{eq:sum_low_bound}
L_i \frac{S_i}{c} \leq V_i(t) \leq \frac{\bar{L}_i}{c} \sum_{j=1}^{i}I_j - \sum_{j=1}^{i-1} S_j \frac{L_j}{c} :=V^{\text{max,\mathrm{H}}}_i,
\end{align}
\noindent where $\bar{L}_i = \max(L_{BE}, L_{>i})$, and $ L_{>i}$ is the maximum packet length of the classes having less priority than class $i$.
These bounds are formally proven in \cite{daigmorte_modelling_2018}.

When $p=2$ we can compare the J- and H-bounds. For class 1 the bounds are identical, i.e., $V^{\text{max,\mathrm{J}}}_1=V^{\text{max,\mathrm{H}}}_1$. For class $2$, whenever $L_2<L_{BE}$, which is often but not always assumed, we have $V^{\text{max},\mathrm{J}}_2<V^{\text{max},\mathrm{H}}_2$; otherwise it may happen that $V^{\text{max},\mathrm{J}}_2>V^{\text{max},\mathrm{H}}_2$ for some values of the system parameters.

We use the following result, proven as Theorem 7 in \cite{daigmorte_modelling_2018}; for $i=1...p$:
\begin{equation}\label{eq:sum_up_bound}
 \sum_{j=1}^{i}V_j(t) \leq \frac{\bar{L}_i}{c} \sum_{j=1}^{i}I_j.
\end{equation}

\section{Improved Upper Bound on the Credit of an arbitrary AVB class}
\label{sec:bound}
\begin{theorem}[Improved Credit Bounds] \label{thm:bound}
The credit of an AVB class $i$, $V_i(t)$, is upper bounded, $\forall t \geq 0$, by:
\begin{align}\label{eq:defcreditUp}
V_i^{\mathrm{max}}=\frac{I_i}{c(c-\sum_{j=1}^{i-1}I_j)}\Bigg(c \bar{L}_i -\sum_{j=1}^{i-1} S_j L_j\Bigg).
\end{align}
\end{theorem}

\begin{proof}
Consider a time $t \geq 0$ and define time instant $s = \sup\{u\in[0,t]: V_i(u)=0\}$. Based on the definition of $s$: 
First, $V_i(u) \neq 0,~\forall u \in (s,t]$. This implies no credit reset during $[s,t]$, i.e., $V_i(\cdot)$ is continuous during this interval. Second, CDT either finishes a transmission at $s$ or is not transmitting at $s$. Indeed, otherwise, since $V_i(s)=0$,  and the credit of $i$ is frozen during the transmission of CDT, it would be true that $V_i(s^+)=0$ and thus, $s \neq \sup\{u\in[0,t]: V_i(u)=0\}$.
	
If $V_i(t)=0$, then $s=t$ and the result is straightforward. 
Since $V_i()$ is continuous in $[s,t]$, and $V_i(u) \neq 0, \forall u \in (s,t]$, then, either $\forall u\in(s,t]: V_i(u) > 0$ or $\forall u\in(s,t]: V_i(u) < 0$. If $V_i(t)<0$, then $\forall u\in(s,t]: V_i(u) < 0$ and the theorem is straightforward to show. 
As a result, the rest of the proof focuses on the case $V_i(t)>0$, i.e., $\forall u\in(s,t]: V_i(u) > 0$. The class $i$ cannot start a transmission at time $s$, otherwise, by rule R2 since $V_i(s)=0$, its credit would decrease to negative values, which contradicts our assumption that $\forall u\in(s,t]: V_i(u) > 0$. Note that since the credit of class $i$ is positive in $(s,t]$, its backlog is also positive in $(s,t]$.

Since $\forall u\in(s,t]: V_i(u) > 0$ and due to rule R1, a class with lower priority than $i$ cannot start a transmission in $(s,t]$. However, in order to consider non-preemptive AVB and BE classes, we should account for the case that a lower priority class has initiated a transmission the latest on $s$ and
is still transmitting at $s$. To do so, we define the time instant $t_0$, with $s\leq t_0 $, as the end of the transmission of the residual of a lower priority packet after time $s$. The latter is denoted by $l^{LO} \leq \bar{L}_i$. If there is no transmission of a lower priority packet, then $l^{LO}=0$. If CDT is preemptive then the transmission of $l^{LO}$ may be interrupted and re-continued. Let $d^{L0}$ be the aggregated time period that the credit is frozen within $[s,t_0]$. Then, $t_0= s+d^{L0}+\frac{l^{LO}}{c}$ and we assume $t_0 \leq t$.

The interval $[t_0,t]$ can be split into a sequence of sub-intervals during which class $i$ alternates between non-transmission and transmission.  Let $[t_0,t_1], [t_1,t_2],...,[t_{n-1}, t_{n}]$ be such a sequence, with $t_0 \leq t_1 < ... < t_n = t$. We allow $t_0=t_1$ as this makes it possible to assume that class $i$ does not transmit in the first interval $[t_0, t_1]$  (i.e., if class $i$ starts transmission at time $t_0$ we set $t_1=t_0$). It follows that for the
even intervals $[t_k,t_{k+1}]$, with $k\in\{0,2,4,...,2\floor{\frac{n}{2}}\}$, we have $\frac{d}{dt} V_i(u) \geq 0$, $\forall u\in(t_k,t_{k+1})$. Indeed, during non transmission, the credit either increases or remains constant, by rules R3 and R4. Conversely, for the odd  intervals $[t_k,t_{k+1}]$, with  $k\in\{1,3,5,...,2\ceil{\frac{n}{2}}-1\}$, we have $\frac{d}{dt} V_i(u) < 0$, $\forall u\in(t_k,t_{k+1})$. 

%

	
 Let us define as $d_k$ the aggregated time period that the credit is frozen within the even interval $[t_k,t_{k+1}]$. Next, we study the credit variation for all classes, starting with the interval $[s,t_0]$, then following with even and odd intervals in $(t_0,t]$. In $[s,t_0]$: 
\begin{itemize}
   \item Each class $j<i$ gains credit if it has backlog or negative credit (rule R3), except if CDT transmits, i.e.,
	\begin{align}
		V_j(t_0)-V_j(s)\leq I_j(t_0-s).
	\end{align}
	By summing up for all $j < i$, we have:
	\begin{align}\label{eq:credit_j_s_t_0}
	\sum_{j=1}^{i-1}\Big(V_j(t_0)-V_j(s)\Big)\leq  \sum_{j=1}^{i-1}I_j (t_0-s).
	\end{align}
\item Class $i$ gains credit because it has backlog as explained above, except if CDT transmits
	\begin{align}
	V_i(t_0)-V_i(s)=I_i(t_0-s) -I_i d^{LO},
	\end{align}
	and since $V_i(s) = 0$ and $t_0= s+d^{LO}+\frac{l^{LO}}{c} $, we get:
	\begin{equation}\label{eq:credit_i_s_t_0}
	V_i(t_0)=I_i(t_0-s-d^{LO}) \leq I_i \frac{l^{LO}}{c}\leq I_i \frac{\bar{L}_i}{c}.
	\end{equation}
\end{itemize}
	 For the odd intervals, $[t_{2k-1},t_{2k}],~ (1 \leq k \leq \floor{\frac{n}{2}})$, we have:\begin{itemize}
   \item
	  Since the credit of class $i$ reduces, the higher priority classes do not transmit within $[t_{2k-1},t_{2k}]$ and $\forall j<i: V_j(t_{2k-1})\leq 0$. They gain credit if they have positive backlog or negative credit, i.e.,
	\begin{align}\label{eq:credit_j_odd}
		V_j(t_{2k})-V_j(t_{2k-1}) \leq I_j(t_{2k}-t_{2k-1}).
	\end{align}
Summing them up for all $j < i$:
\begin{align}
\sum_{j=1}^{i-1}\Big(V_j(t_{2k})-V_j(t_{2k-1})\Big)\leq  \sum_{j=1}^{i-1}I_j (t_{2k}-t_{2k-1}).
\end{align}\label{eq:credit_i_odd}
	\item The credit of class $i$ reduces due to transmission (R2):
	\begin{align}
		V_i(t_{2k})-V_i(t_{2k-1})=S_i(t_{2k}-t_{2k-1}).
	\end{align}
\end{itemize}
	 In even intervals $[t_{2k},t_{2k+1}],~ (0 \leq k \leq \floor{\frac{n-1}{2}})$, we have:
\begin{itemize}
   \item
   There should exist an AVB class $j<i$ that transmits or all AVB and BE classes wait for CDT (for an aggregated time $d_{2k}$). 
	Define $a_{j,{2k}}$ as the aggregated period of time that class $j$ transmits packets in $[t_{2k},t_{2k+1}]$. Then, by using that $I_j - S_j = c$, we obtain,
	\begin{align}
	V_j(t_{2k+1})-V_j(t_{2k}) \leq  I_j (t_{2k+1}-t_{2k})- ca_{j,2k}-I_j d_{2k}.
	\end{align}
	Summing up for all $j < i$, and considering that $t_{2k+1} - t_{2k} = d_{2k}+ \sum_{j=1}^{i-1}a_{j,{2k}}$, we obtain, 
	\begin{align}\nonumber
	&\sum_{j=1}^{i-1}\Big(V_j(t_{2k+1})-V_j(t_{2k})\Big)  \\  
	& \leq- \Big(c -\sum_{j=1}^{i-1}I_j\Big) (t_{2k+1}-t_{2k})+ (c-\sum_{j=1}^{i-1}I_j) d_{2k}.
	\end{align}
	\item 
	The credit of class $i$ increases or is frozen for an aggregated time $d_{2k}$, i.e.,
	\begin{align}
		V_i(t_{2k+1})-V_i(t_{2k})=I_i(t_{2k+1}-t_{2k})- I_i d_{2k}.
	\end{align}
\end{itemize}	
	Next, we study the credit variation within $[t_0,t_n]$. First we assume that $n$ is odd. 
	By summing up the credit variations for all intervals and all classes $j <i$, we have:
	\begin{align}\nonumber
	&\sum_{j=1}^{i-1}\Bigg[\Big(V_j(t_1)-V_j(t_0)\Big)+\Big(V_j(t_{2})-V_j(t_{1})\Big)+... \\ \ \nonumber
	&+ \Big(V_j(t_{n-1})-V_j(t_{n-2})\Big)+ \Big(V_j(t_{n})-V_j(t_{n-1})\Big)\Bigg]\\\nonumber
	\leq &- \Big(c -\sum_{j=1}^{i-1}I_j\Big) (t_1 - t_0)+...+\sum_{j=1}^{i-1}I_j (t_{n-1}-t_{n-2})\\ 
	& - \Big(c -\sum_{j=1}^{i-1}I_j\Big) (t_n - t_{n-1}) +\Big(c -\sum_{j=1}^{i-1}I_j\Big) \sum_{k=0}^{k=\lfloor \frac{n}{2} \rfloor} d_{2k}.
	\end{align}
	Therefore, by setting $\alpha = (t_{2}-t_{1})+(t_{4}-t_{3})+...+(t_{n-1}-t_{n-2})$ and $\Delta t=t_n - t_0 -\sum_{k=0}^{k=\lfloor \frac{n}{2} \rfloor}  d_{2k}$, we can write,
	\begin{align}
	 \sum_{j=1}^{i-1} \Big(V_j&(t_n)-V_j(t_0)\Big) \leq
	-(c-\sum_{j=1}^{i-1}I_j) \Delta t+c \alpha.\label{eq:lp_main_con_gen}
		\end{align}
	Next, by summing up the credit variations for all intervals for class $i$ and considering $S_i = I_i-c$,
	\begin{align}
	&V_i(t_{n}) -V_i(t_0) = I_i(t_1-t_0)+(I_i-c)(t_{2}-t_{1})+... \nonumber \\
	&+(I_i-c)(t_{n-1}-t_{n-2}) +I_i(t_{n}-t_{n-1})-I_i \sum_{k=0}^{k=\lfloor \frac{n}{2} \rfloor}  d_{2k}\nonumber\\
	&=I_i\Delta t-c \alpha. \label{eq:lp_main_obj_gen_sim}
	\end{align}
	By Eq. \eqref{eq:sum_up_bound}, we obtain
	\begin{align}
	\sum_{j=1}^{i-1}V_j(t_0) &\leq 
	-V_i(t_0)+\frac{\bar{L}_i}{c}I_i+\frac{\bar{L}_i}{c} \sum_{j=1}^{i-1}I_j.\label{eq:Vjlbound}
	\end{align}
	We lower bound the left hand-side of Eq. \eqref{eq:lp_main_con_gen} using the lower bound of Eq. \eqref{eq:sum_low_bound} and the bound of Eq. \eqref{eq:Vjlbound}; therefore,
	\begin{align} \label{eq:lp_main_con_gen_sim}
	&V_i(t_0)-K \leq c \alpha -(c- \sum_{j=1}^{i-1}I_j) \Delta t,
	\end{align}
	where $K = -\sum_{j=1}^{i-1} L_j \frac{S_j}{c}+\frac{\bar{L}_i}{c}I_i+ \frac{\bar{L}_i}{c}\sum_{j=1}^{i-1}I_j \geq 0$. 
	
	\noindent Eq. \eqref{eq:lp_main_con_gen_sim} gives an upper bound on $\Delta t$, i.e.,
	\begin{align} \label{eq:upper_bound_dt}
	\Delta t \leq \frac{c \alpha+K-V_i(t_0)}{c- \sum_{j=1}^{i-1}I_j}.
	\end{align}
	\noindent By using Eq. \eqref{eq:upper_bound_dt} in Eq. \eqref{eq:lp_main_obj_gen_sim}, we obtain:
	\begin{align}
	V_i(t_n)& \leq I_i\Bigg(\frac{c \alpha+K-V_i(t_0)}{c- \sum_{j=1}^{i-1}I_j} \Bigg)-c \alpha +V_i(t_0) \nonumber\\
	&=I_i\Bigg(\frac{K}{c- \sum_{j=1}^{i-1}I_j}\Bigg) +V_i(t_0)\Bigg(\frac{c- \sum_{j=1}^{i-1}I_j-I_i}{c- \sum_{j=1}^{i-1}I_j}\Bigg) \nonumber \\&-c\alpha\Bigg(\frac{c- \sum_{j=1}^{i-1}I_j-I_i}{c- \sum_{j=1}^{i-1}I_j}\Bigg).
	\end{align}
	Next, considering $\sum_{j=1}^{i}I_j < c$, and since by Eq. \eqref{eq:credit_i_s_t_0} $V_i(t_0) \leq \frac{\bar{L}_i}{c}I_i$, we obtain:
	\begin{align}
	V_i(t_n) \leq
	  \frac{I_i}{c\Big(c- \sum_{j=1}^{i-1}I_j\Big)}\Bigg(cK+\bar{L}_i\Big(c- \sum_{j=1}^{i-1}I_j-I_i\Big)\Bigg).
	\end{align}
	By replacing the value of $K$, the credit of class $i$ at time $t_n$, where $n$ is odd, is upper bounded by $V_i^{\mathrm{max}}$ given in the statement. %
If $n$ is even, then:
	\begin{equation}
	V_i(t_n) = V_i(t_{n-1}) + S_i(t_n-t_{n-1}).
	\end{equation}
	Since $S_i(t_n-t_{n-1}) \leq 0$, it is true that $V_i(t_n) \leq V_i(t_{n-1})$.
	As $n$ is even, $n-1$ is odd. We have already found a bound for $t_k$ when $k$ is odd, which is $V_i^{\mathrm{max}}$. 
	Since $t=t_n$, and $n$ is either odd or even, then $V_i(t) \leq V_i^{\mathrm{max}}, ~\forall t\geq 0$ and this completes the proof.
	\end{proof}
\begin{proposition}
The credit bound, $V_i^{\max}$, given in \eqref{eq:defcreditUp} is tight for the two highest priority classes, i.e., for each set of parameter values and each class 1,2, there is a scenario for which the credit counter attains the bound.
\end{proposition}
\begin{proof}
The credit of class 1 achieves the value $V_1^{\max}$ in the following scenario. 
Assume that all queues have zero backlog. Just before the backlog of class 1 becomes positive, there is an arrival of a lower priority class packet with length $\bar{L}_1$. This packet starts being transmitted according to R1 (assuming the gate is open for it). It takes $\frac{\bar{L}_1}{c}$ to transmit the lower priority packet. 
During the transmission, since class 1 has positive backlog it gains credit according to the rule R3. At the end of the transmission, the credit of class 1 becomes $I_1 \frac{\bar{L}_1}{c}$ i.e., equal to $V_1^{\max}$ (Eq. \eqref{eq:defcreditUp}).
	
The tightness scenario for class 2 is as follows. Assume that all queues for all classes have zero backlog. 
Just before the backlog of class 2 becomes positive, there is an arrival of a lower priority class packet with length $\bar{L}_2$. This packet starts being transmitted (rule R1) since at this moment there is no packet of class 1, 2 and the gate is open. Just after the transmission, the backlog of class 1 becomes also positive. 
Due to positive backlog, the classes 1,2 gain credit according to the rule R3. At the end of transmission of the lower priority packet, the credit values of classes 1 and 2 are 
$I_1 \frac{\bar{L}_2}{c}$ and $I_2 \frac{\bar{L}_2}{c}$, respectively. Then, class 1 starts transmission for a time interval $\frac{I_1}{-S_1} \frac{\bar{L}_2}{c}$ until its credit becomes zero. During the latter transmission, class 2 gains credit of $I_2\frac{I_1}{-S_1} \frac{\bar{L}_2}{c}$. 
When the credit value of class 1 is zero, it transmits a packet with maximum length $L_1$ for a time interval $\frac{L_1}{c}$, during which class 2 gains credit equal to $I_2\frac{L_1}{c}$.
The total credit gained by class 2 is $I_2 \frac{\bar{L}_2}{c} + I_2\frac{I_1}{-S_1} \frac{\bar{L}_2}{c} + I_2\frac{L_1}{c} = \frac{I_2}{c (c-I_1)}(c\bar{L}_2 - S_1L_1)$, which is equal to Eq. \eqref{eq:defcreditUp}.
\end{proof}
We now formally compare $V_i^{\max}$ with the existing J- and H-bounds and show that our bounds improve on all existing bounds.

\begin{proposition}
\begin{enumerate}
  \item $V_1^{\mathrm{max}} =  V_1^{\mathrm{max},\mathrm{J}}=V_1^{\mathrm{max},\mathrm{H}}$.
  \item  $V_2^{\mathrm{max}} \leq V_2^{\mathrm{max},\mathrm{J}}$ and the inequality is strict if $L_2>L_{BE}$.
  \item For $j=2,...p$, $V_j^{\mathrm{max}} < V_j^{\mathrm{max},\mathrm{H}}$.
\end{enumerate}
\end{proposition}
\begin{proof}1) is straightforward. For 2), observe that $V_2^{\mathrm{max}}$ can be obtained by replacing $\bar{L}$ with $L_{BE}$ in $V_2^{\mathrm{max},\mathrm{J}}$. For 3), after some algebra we find:
\begin{equation}  
V_i^{\mathrm{max},\mathrm{H}} - V_i^{\mathrm{max}}
= \frac{c-\sum_{j=1}^{i}I_j}{c\Big(c-\sum_{j=1}^{i-1}I_j\Big)} \Big(\bar{L}_i\sum_{j=1}^{i-1}I_j-\sum_{j=1}^{i-1} S_jL_j\Big).\label{eq:bounddif1}
\end{equation}

By hypothesis, $c > \sum_{j=1}^{i}I_j$. Since $I_j>0$, $S_j<0$ and $i\geq 2$, it follows that  $\bar{L}_i\sum_{j=1}^{i-1}I_j-\sum_{j=1}^{i-1} S_jL^j > 0$, thus the last term of Eq. \eqref{eq:bounddif1} is strictly positive. 
\end{proof}

\section{Numerical Evaluation}
\label{sec:eval}
Consider a TSN scheduler with one CDT, three AVB and one BE classes, that is connected to a link with line rate $c=100$ Mbps. Assume $I_1$, $I_2$, and $I_3$ are $50\%$, $15\%$, and $10\%$ of the link rate, and for any AVB class $i$, $S_i=I_i - c$. Also, 
$L_1 = 0.2$KB, $L_2 = 1.5$KB, and $L_3=0.5$KB; 
$L_{BE}=1$KB.
The CDT is is constrained by an affine arrival curve $a(t)=r t+b$\cite{le_boudec_network_2001} with parameters $r=12.8$Kbps and $b=1.6$Kb.

The credit upper bounds for the three AVB classes computed by J-bounds \eqref{eq:boyer_avb_A}, H-bounds \eqref{eq:sum_low_bound}, and Theorem \ref{thm:bound} are shown in Table \ref{table:credit_num}.
As we know, the bounds coincide for class 1. In contrast, for class 2, our new credit bound is less than the J-bound by $18.5\%$\footnote{As aforementioned, J-bounds exist only for the case $p=2$. Thus, in the computation of the J-bound, class $3$ is treated as a BE class.} and than the H-bound by $56\%$. For class 3 the bound by Theorem 1 is less than the H-bound by $68.1\%$, while the J-bound does not exist for class 3.
\begin{table}[]
	\centering
	\caption{Credit upper bounds of three AVB classes obtained by \cite{azua_complete_2014} ($V_i^{\mathrm{max},\mathrm{J}}$), \cite{daigmorte_modelling_2018} ($V_i^{\mathrm{max},\mathrm{H}}$), and Theorem  \ref{thm:bound} ($V_i^{\mathrm{max}}$).}
	\begin{tabular}{|c|c|c|c|l}
		\cline{1-4}
		& $i=1$ & $i=2$ & $i=3$ &  \\ \cline{1-4}
	$V_i^{\mathrm{max}}$ (Kb)& $6$    & $2.64$  &    $5.43$     &  \\ \cline{1-4}
	$V_i^{\mathrm{max},J}$ (Kb)  & $6$     & $3.24$  &   -      &  \\ \cline{1-4}
	$V_i^{\mathrm{max},H}$ (Kb) & $6$     & $6$     &   $17$      &  \\ \cline{1-4}
	\end{tabular}
\label{table:credit_num}
\end{table}


As discussed, the credit upper bound has an impact on the latency bound of a FIFO system and subsequently of the end-to-end latency, as shown in \cite{mohammadpour_latency_2018}. This can be seen by the improvement in the latency term of the service curves provided to the AVB classes, analyzed below. 
According to Eq. (22) of \cite{TSN_TR} (that is the companion paper of \cite{mohammadpour_latency_2018}), a service curve for the AVB class $i$ is:
\begin{align}
\beta_i(t) =\frac{(c-r) I_i}{I_i - S_i} \Big[t -\frac{c V_i^{\mathrm{M}}}{(c-r)I_i} + \frac{b+\frac{rL^{\text{N}}}{c}}{c-r}\Big]^+,
\end{align}
where $V_i^{\mathrm{M}}$ is a credit bound for class $i$ and $L^{\text{N}}$ is the maximum packet length of all classes except CDT. 
For class 2, the service curve latency in microseconds is $192.02$, $232.02$ and $416.05$ if computed with $V_2^{\mathrm{max}}$, $V_2^{\mathrm{max, \mathrm{J}}}$ and $V_2^{\mathrm{max,H}}$, respectively. Thus, Theorem 1 improves the service curve latency of class 2 by $17\%$ compared with \cite{azua_complete_2014} and by $53.8\%$ compared with \cite{daigmorte_modelling_2018}. For class 3, the service curve latency in microseconds is $558.93$ and $1716.22$ if computed with $V_3^{\mathrm{max}}$ and $V_3^{\mathrm{max,H}}$, respectively. Thus, Theorem~1 improves the service curve latency of class 3 by $67.4\%$ compared with \cite{daigmorte_modelling_2018}.



\bibliographystyle{IEEEtran}
\bibliography{ref}


\end{document}